\documentclass[a4paper]{llncs}
\usepackage[centertags]{amsmath} 
\usepackage{amssymb}
\newcommand{\val}{\mathop{\rm val}} 
\newcounter{enumroman}
\newenvironment{romanitems}{\begin{list}{\bfseries(\roman{enumroman})\hfill}{\usecounter{enumroman}
\setlength{\labelwidth}{\leftmargin}\addtolength{\labelwidth}{-1\labelsep}
\topsep=0mm plus 2pt\itemsep=0mm plus 1pt\parsep=0mm\itemsep=0mm plus
1pt\itemindent=0mm}}{\end{list}}

\title{           A Generalized Markov-Chain Modelling
\\
                  Approach to $(1,\lambda )$-ES Linear Optimization}
\author{          Martin Hole\v{n}a\,\inst{1}\and Alexandre Chotard\,\inst{2}} 
\institute{       Institute of Computer Science, Academy of Sciences,
\\
                  Pod vod\'arenskou v\v{e}\v{z}\'{\i} 2, Prague, Czech Republic,
                  \email{martin@cs.cas.cz}\and INRIA Saclay-Ile-de-France, LRI,
                  University Paris-Sud, France}

\begin{document}
\bibliographystyle{splncs}

\maketitle

                  


\section{         Basic Setting}

\begin{definition} 
                  Let $d,\lambda\! \in \mathbb{N} ,d \ge 2, \eta: \{
                  (r^{i,j})^{i=1\dots\lambda}_{j\in \mathbb{N}
                  }:r^{i,j}\in\mathbb{R}, i=1\dots\lambda,j\in \mathbb{N}
                  \}\to\mathbb{R}_+ $ be Borel-measurable, and let $n\in
                  \mathbb{R}^d$ fulfils
\begin{gather}
                  [n]_1,[n]_2>0\;\&\;d\ge 3 \Rightarrow [n]_k=0 \text{ for } k\ge 3,
\end{gather}      
                  with $[]_i$ standing for the $i$-th component. Denote
                  $\mathcal{B}(\mathbb{R}_+)$ the Borel $ \sigma$-algebra on
                  $\mathbb{R}_+$, and consider a system of $d$-dimensional
                  i.i.d. random vectors
                  $(M_t^{i,j})_{j,t\in\mathbb{N}}^{i=1\dots\lambda }$, a
                  system of random variables $(\Sigma_t)_{t=0}^\infty$ and
                  systems of $d$-dimensional random vectors $(X_t)_{t\in
                  \mathbb{N}_0}, (Y_t^i)_{t\in\mathbb{N}}^{i=1\dots\lambda }$
                  such that:
\begin{romanitems}
\item             $X_0$ and $\Sigma_0$ are independent of each other, as well as of
                  $(M_t^{i,j})_{t,j\in\mathbb{N}}^{i=1\dots\lambda }$,
\begin{gather} 
\label{gx0}        
                  P (\Sigma_0>0)= P (-n^\top X_0>0)=1,
\end{gather}      
\item             for $t\in\mathbb{N}$, 
\begin{gather} 
\label{gst}        
                  \Sigma_t=\eta
                  ((M_t^{i,j})_{j\in\mathbb{N}}^{i=1\dots\lambda })
                  \Sigma_{t-1}, 
\end{gather}      
\item             for $t\in\mathbb{N},i=1\dots\lambda$,
\begin{gather} 
\label{gyt}        
                  Y_t^i=X_{t-1}+\Sigma_tM_t^{i,j_{t,i} (-n^\top X_{t-1}) },
\\
                  \text{where } j_{t,i} (\delta ) =\min\{j:n^\top
                  M_t^{i,j}<\delta \},\delta>0,     
\end{gather}       
\item             for $t\in\mathbb{N}$,
\begin{gather} 
\label{gxt}        
                  X_t=Y_t^{i_t},\text{with
                  }i_t=\min\{\arg\max_{i=1\dots\lambda}[Y_t^i]_1\}. 
\end{gather}      
\end{romanitems}   
                  Then:
\begin{enumerate}
\item[a)]         $Y_t^i,t\in\mathbb{N},i=1\dots\lambda$ is called
                  \emph{$i$-th individual} or \emph{$i$-th feasible solution} in
                  \emph{generation $t$}, with \emph{movement $M_t^{i,j_{t,i}}$}
                  and \emph{step size $\Sigma_t$}.
\item[b)]         $X_t,t\in\mathbb{N}$ is called the \emph{best solution} among
                  the  \emph{population of size $\lambda $} in \emph{generation $t$}. 
\end{enumerate}   
\end{definition}

\begin{proposition} 
                  For $t\in\mathbb{N}$, define 
\begin{gather} 
\label{gdt}        
                  D_t=\frac{-n^\top X_{t-1}}{\Sigma_t}.
\end{gather}      
                  Then no matter what distribution the movements
                  $M_t^{i,j},i=1\dots\lambda,j,t\in\mathbb{N}$ have, the
                  following holds: 
\begin{enumerate}
\item             $(D_t,\Sigma_t)_{t\in\mathbb{N}}$ is a homogeneous Markov chain;
\item             if $\Sigma_t=\sigma\in\mathbb{R}_+ a.s.$ for $t\in\mathbb{N}$
                  (e.g., if $\Sigma_0$ has a degenerate distribution with $
                  \Sigma_0=\sigma$ a.s., and $\eta$ is identity), even
                  $(D_t)_{t\in\mathbb{N}}$ is a homogeneous Markov chain,
\end{enumerate}   
\end{proposition}    
\begin{proof}
                  Due to (\ref{gst})--(\ref{gdt}),
                  $(D_t,\Sigma_t)_{t\in\mathbb{N}}$ can be considered as an
                  $\mathbb{R}_+\times\mathbb{R}_+$-valued random process.

                  The fact that
                  $(M_t^{i,j})_{j,t\in\mathbb{N}}^{i=1\dots\lambda }$ are
                  i.i.d. implies that also $(n^\top M_t^{i,j}
                  )_{t,j\in\mathbb{N}}^{i=1\dots\lambda }$ and $(\eta
                  ((M_{t+1}^{i,j})_{j\in\mathbb{N}}^{i=1\dots\lambda
                  }))_{t\in\mathbb{N}}$ are i.i.d., and for each $\delta>0$:
\begin{itemize}
\item             $(j_{t,i} (\delta ))_{t\in\mathbb{N}}^{i=1\dots\lambda }$
                  are i.i.d.,
\item             $([M_t^{i,j_{t,i} (\delta )}]_i)_{t\in\mathbb{N}}^{i=1\dots\lambda }$ 
                  are i.i.d.,
\item             $(\min\{\arg\max_{i=1\dots\lambda}[M_t^{i,j_{t,i} (\delta)}]_1
                  \})_{t\in\mathbb{N}}$ are i.i.d.,
\item             $(M_t^{i_t,j_{t,i_t} (\delta )})_{t\in\mathbb{N}}$ are i.i.d.
\end{itemize}     
                  Let now $t\in\mathbb{N},t\ge
                  2,\delta_0,\sigma_0\dots\delta_{t-1},\sigma_{t-1},\delta,\sigma
                  \in\mathbb{R}_+ $ and $A,B \in \mathcal{B}(\mathbb{R}_+) $. From
                  (\ref{gx0})--(\ref{gxt}) follows:
\begin{multline} 
\label{gmc}        
                  P ((D_{t+1},\Sigma_{t+1})\in A\times
                  B|(D_0,\Sigma_0)=
\\
                  =(\delta_0,\sigma_0)\dots(D_{t-1},\Sigma_{t-1})=
                  (\delta_{t-1},\sigma_{t-1}),(D_t,\Sigma_t)=(\delta,\sigma) )=
\\
                  =P (M_{t+1}^{i_{t+1},j_{t+1,i_{t+1}} (\delta )}\!\!\in\!\{\xi\!\in\!
                  \mathbb{R}:\xi+\delta \!\in\! A\}\&\;\eta
                  (M_{t+1}^{i,j})_{j\!\in\!\mathbb{N}}^{i=1\dots\lambda
                  })\!\in\!\{\varsigma\!\in\!\mathbb{R}_+ :\varsigma\sigma\!\in\! B\} )
\\
                  =P ((D_{t+1},\Sigma_{t+1})\in A\times
                  B|(D_t,\Sigma_t)=(\delta,\sigma)).                   
\end{multline}  
                  Therefore, $(D_t,\Sigma_t)_{t\in\mathbb{N}}$ is a Markov
                  chain. Its homogeneity is a consequence of $(n^\top M_t^{i,j}
                  n)_{t,j\in\mathbb{N}}^{i=1\dots\lambda }$ and $(\eta
                  (M_{t+1}^{i,j})_{j\in\mathbb{N}}^{i=1\dots\lambda
                  }))_{t\in\mathbb{N}}$ being i.i.d. Moreover, if $(\forall
                  t\in\mathbb{N}) \Sigma_t=\sigma\in\mathbb{R}_+ a.s.$, then for
                  $A \in \mathcal{B}(\mathbb{R}_+)$ such that $\sigma\in A$,
\begin{gather}
                  P ((D_{t+1},\Sigma_{t+1})\in A\times\{\sigma\})=P ((D_{t+1})\in A),
\end{gather}      
                  which together with (\ref{gmc}) implies that also
                  $(D_t)_{t\in\mathbb{N}}$ is a homogeneous Markov chain. \qed  
\end{proof}

\section{          Markov Chain Properties if Step Size Is Constant}

                  Following the first part of \cite{chotard14markov}, we
                  restrict attention to the constant step size in the remainder
                  of the paper. Proposition~\ref{g2} below generalizes
                  Lemmas~1 and 3 and Propositions~3--4 from \cite{chotard14markov}.

\begin{proposition} 
\label{g2}        
                  Denote $\mu_+$ the Lebesgue measure on
                  $(\mathbb{R}_+,\mathcal{B}(\mathbb{R}_+))$, $H$ the
                  distribution function corresponding to the first two
                  dimensions of the random vectors
                  $M_t^{i,j},j,t\in\mathbb{N},i=1\dots\lambda$, and $H_\delta
                  $ and $H^\star_\delta $ for $\delta>0$ the distribution
                  functions corresponding to the first two dimensions of the
                  random vectors $M_t^{i,j_{t,i} (\delta
                  )},t\in\mathbb{N},i=1\dots\lambda$ and $M_t^{i_t,j_{t,i_t}
                  (\delta )},t\in\mathbb{N}$, respectively. Let $\lambda \ge 2,
                  (\forall t\in\mathbb{N}) \Sigma_t=\sigma\in\mathbb{R}_+ a.s.$,
                  and the distribution $H$ be absolutely continuous, with
                  continuous strictly positive density $h$. Finally, define the
                  open linear half-space $L_a $ for $a\in\mathbb{R}$ and the
                  functions $\mathbb{I}_A $ and $V_\alpha$ by
\begin{align} 
\label{ghl}        
                  L_a&=\{x\in\mathbb{R}^2:[x]_1[n]_1+[x]_2[n]_2<a  \}
                  \text{ for } a\in\mathbb{R},
\\ 
\label{g1a}        
                  (\forall x\in\mathbb{R}^2)\; \mathbb{I}_A (x) &=
\begin{cases}
                  1&\text{ if } x\in A,
\\
                  0&\text{ if } x\not\in A,
\end{cases}       
                  \text{ for } A \subset \mathbb{R}^2.
\\ 
\label{gva}        
                  (\forall \delta ) V_\alpha (\delta )&=\exp (\alpha
                  \delta)\text{ for } \alpha >0.
\end{align} 
                  Then the following holds:
\begin{enumerate}
\item             The distribution functions $H_\delta $ and
                  $H^\star_\delta $ are for each $\delta>0$
                  absolutely continuous, and their densities are, respectively,
\begin{align} 
\label{gdd}        
                  (\forall x\in\mathbb{R}^2)\; h_\delta (x) &=\frac{ h (x)
                  \mathbb{I}_{L_\delta}(x)}{ H (L_\delta) } \text{ and}   
\\ 
\label{gds}        
                  (\forall x\in\mathbb{R}^2)\;h^\star_\delta(x)  &=\lambda h_{
                  \delta }(x)H_{ \delta } ((-\infty,[x]_1
                  )\times\mathbb{R})^{\lambda-1}. 
\end{align} 
\item             The transition probability kernel of $(D_t)_{t\in\mathbb{N}}$, defined 
\begin{gather} 
\label{gpk}        
                  (\forall \delta \in\mathbb{R}_+) (\forall A \in
                  \mathcal{B}(\mathbb{R}_+)\;P (\delta,A)=P (D_{t+1}\in
                  A|D_t=\delta )
\end{gather}      
                  is with respect to its first argument $\delta $ continuous
                  on $\mathbb{R}_+$.
\item             $(D_t)_{t\in\mathbb{N}}$ is $\mu_+$-irreducible and aperiodic.
\item             If the following two additional conditions are fulfilled:
\begin{multline} 
\label{gmh}        
                  (\exists \varepsilon>0) (\forall \alpha \in (0,\varepsilon ) )
                  (\forall k\in\mathbb{N}) \;\int_{\mathbb{R}^2} |\alpha n^\top
                  x|^kh (x)dx<+\infty \;\&
\\
                  \&\;\int_{\mathbb{R}^2} \exp(|\alpha n^\top x|)h
                  (x)dx=\sum^{+\infty}_{k\in\mathbb{N}}\int_{\mathbb{R}^2} |\alpha
                  n^\top x|^kh (x)dx<+\infty, 
\end{multline}      
\begin{gather} 
\label{gal}        
                  \text{and } \int_{\mathbb{R}^2}n^\top xH((-\infty,[x]_1)\times
                  \mathbb{R} )^{ \lambda-1}h (x) dx>0, 
\end{gather}      
                  then $(\exists \alpha_n>0)$ such that for $\alpha \in (0,\alpha_n)$
                  is $(D_t)_{t\in\mathbb{N}}$ $V_\alpha-$geometrically ergodic and
                  positive Harris recurrent.
\end{enumerate}   
\end{proposition} 
\begin{proof} 
\begin{enumerate}
\item             Due to (\ref{ghl}) and the strict positivity of $h$,
\begin{gather} 
\label{gl0}        
                  H (L_\delta)>H (L_0)>0
\end{gather}      
                  no matter which $\delta>0 $ is considered.  Let now
                  $\delta>0,A\in\mathcal{B}(\mathbb{R}^2)$. Then for
                  $t\in\mathbb{N},i=1\dots\lambda$,
\begin{multline}
                  H_\delta (A)= P (M_t^{i,j_{t,i} (\delta )}\in A)=\sum_{k\in\mathbb{N}}P
                  (M_t^{i,j_{t,i} (\delta )}\in A|j_{t,i}=k)P (j_{t,i}=k) =
\\
                  =\sum_{k\in\mathbb{N}}P (M_t^{i,k}\in A\cap L_\delta\;\&\;
                  (\forall \ell<k)\;M_t^{i,\ell}\not\in L_\delta|M_t^{i,k}\in
                  L_\delta\;\&
\\
                  \&\; (\forall \ell<k)\;M_t^{i,\ell}\not\in L_\delta)P
                  (j_{t,i}=k) =
\\
                  =\sum_{k\in\mathbb{N}} \frac{P (M_t^{i,k}\in A\cap
                  L_\delta)\prod_{\ell<k}P (M_t^{i,\ell}\not\in L_\delta) }{P
                  (M_t^{i,k}\in L_\delta)\prod_{\ell<k}P (M_t^{i,\ell}\not\in
                  L_\delta)}P (j_{t,i}=k)=
\\
                  =\sum_{k\in\mathbb{N}} \frac{P (M_t^{i,k}\in A\cap
                  L_\delta)}{P (M_t^{i,k}\in L_\delta)}P (j_{t,i}=k)=\frac{H
                  (A\cap L_\delta) }{H (L_\delta) }=\frac{\int_Ah
                  \mathbb{I}_{L_\delta} }{H (L_\delta)}. 
\end{multline}    
                  Consequently, $H_\delta$ is absolutely continuous, with
                  density $\frac{h\mathbb{I}_{L_\delta}}{H (L_\delta) }$. Due to
                  the absolute continuity of $H_\delta$, the ties between
                  $[M_t^{i,j_{t,i}(\delta )}]_1$ and
                  $[M_t^{k,j_{t,k}(\delta)}]_1$ have probability 0 for any
                  $t\in\mathbb{N},i,k=1\dots\lambda,i\ne k$, thus we can
                  ignore such ties in the subsequent calculations. Let now
                  $x\in\mathbb{R}^2$ be an inner point of $L_\delta$, thus
                  $(\exists \epsilon>0) (\forall
                  x'\in\mathbb{R}^2)\;max(|[x'-x]_1|,|[x'-x]_2|)<\epsilon
                  \Rightarrow x'\in L_\delta$. Then for
                  $t\in\mathbb{N},x\in\mathbb{R}^2,\epsilon_1,\epsilon_2\in
                  (0,\epsilon)$,\small  
\begin{multline} 
                  H^\star_\delta ([x]_1+\epsilon_1,[x]_2+\epsilon_2)-H^\star_\delta
                  (x)\!=\! P (M_t^{i_t,j_{t,i_t} (\delta )}\!\!\in\!
                  ([x]_1,[x]_1+\epsilon_1)\!\times\!
                  ([x]_2,[x]_2+\epsilon_2))
\\
                  =\sum_{k=1}^\lambda P
                  (M_t^{i_t,j_{t,i_t} (\delta )}\in
                  ([x]_1,[x]_1+\epsilon_1)\!\times\!
                  ([x]_2,[x]_2+\epsilon_2)|i_t=k)P (i_t=k)=
\\
                  =\sum_{k=1}^\lambda P (M_t^{k,j_{t,k} (\delta )}\in
                  ([x]_1-\epsilon_1,[x]_1+\epsilon_1)\!\times\!
                  ([x]_2-\epsilon_2,[x]_2+\epsilon_2)\;\&
\\
                  \&\; (\forall
                  i\in\{1\dots\lambda \} \setminus
                  \{k\})[M_t^{i,j_{t,i}(\delta )}]_1 <
                  [M_t^{k,j_{t,k}(\delta)}]_1)P (i_t=k)=   
\\
                  =\sum_{k=1}^\lambda ( P (M_t^{k,j_{t,k} (\delta )}\in
                  ([x]_1,[x]_1+\epsilon_1)\!\times\! ([x]_2,[x]_2+\epsilon_2)\;\&
\\
                  (\forall i\!\in\!\{1\dots\lambda \} \setminus
                  \{k\})[M_t^{i,j_{t,i}(\delta )}]_1\!\! \le\![x]_1\!)+P (M_t^{k,j_{t,k}
                  (\delta )}\!\!\!\in\! ([x]_1,[x]_1\!+\epsilon_1\!)\!\times\!
                  ([x]_2,[x]_2+\epsilon_2)
\\
                  \&\;(\exists\ell\in\{1\dots\lambda \} \setminus
                  \{k\})[x]_1<[M_t^{\ell,j_{t,\ell}(\delta )}]_1 <
                  [M_t^{k,j_{t,k}(\delta)}]_1)\;\&
\\
                  \&\;(\forall i\in\{1\dots\lambda
                  \} \setminus \{k,\ell\})[M_t^{i,j_{t,i}(\delta )}]_1 <
                  [M_t^{\ell,j_{t,i}(\delta )}]_1)P (i_t=k)=
\\
                  =\sum_{k=1}^\lambda P (M_t^{k,j_{t,k} (\delta )}\!\!\in\!
                  ([x]_1,[x]_1+\epsilon_1)\!\times\!
                  ([x]_2,[x]_2+\epsilon_2))\!\!\prod_{i=1,i\ne k}^\lambda P
                  ([M_t^{i,j_{t,i}(\delta )}]_1
                  \le[x]_1)+
\\
                  +\sum_{k=1}^{ \lambda-1}\sum_{\ell=1,\ell\ne k}^\lambda
                  P(M_t^{k,j_{t,k} (\delta )}\in ([x]_1,[x]_1+\epsilon_1)\!\times\!
                  ([x]_2,[x]_2+\epsilon_2)\;\&
\\
                  [M_t^{\ell,j_{t,\ell}(\delta )}]_1\!\!\in\!
                  ([x]_1,[M_t^{k,j_{t,k}(\delta)}]_1)\&(\forall
                  i\!\in\!\{1\dots\lambda \} \setminus
                  \{k,\ell\})[M_t^{i,j_{t,i}(\delta )}]_1\! <\!
                  [M_t^{\ell,j_{t,\ell}(\delta)}]_1) 
\\
                  =\lambda (H_\delta
                  ([x]_1+\epsilon_1,[x]_2+\epsilon_2)-H_\delta (x))H_\delta
                  ((-\infty,[x]_1)\!\times\! \mathbb{R} )^{ \lambda-1}+
\end{multline}  
\begin{multline*}
                  +\sum_{k=1}^{ \lambda-1}\sum_{\ell=1,\ell\ne k}^\lambda
                  P(M_t^{k,j_{t,k} (\delta )}\in ([x]_1,[x]_1+\epsilon_1)\!\times\!
                  ([x]_2,[x]_2+\epsilon_2)\;\&
\\
                  [M_t^{\ell,j_{t,\ell}(\delta)}]_1
                  \!\!\in\! ([x]_1,[M_t^{k,j_{t,k}(\delta)}]_1)\&(\forall
                  i\!\in\!\{1\dots\lambda \} \setminus
                  \{k,\ell\})[M_t^{i,j_{t,i}(\delta )}]_1\! <\!
                  [M_t^{\ell,j_{t,\ell}(\delta)}]_1)
\\
                  \le\lambda (H_\delta
                  ([x]_1+\epsilon_1,[x]_2+\epsilon_2)-H_\delta
                  (x))H_\delta((-\infty,[x]_1)\!\times\! \mathbb{R} )^{
                  \lambda-1}+
\\
                  \sum_{k=1}^{ \lambda-1}\sum_{\ell=1,\ell\ne
                  k}^\lambda\!\! P(M_t^{k,j_{t,k} (\delta )}\!\!\in\!
                  ([x]_1,[x]_1+\epsilon_1)\!\times\!
                  ([x]_2,[x]_2+\epsilon_2)\&[M_t^{\ell,j_{t,\ell}(\delta)}]_1
                  \!\!\in\! ([x]_1,[x]_1+\epsilon_1)
\\
                  =\lambda (H_\delta
                  ([x]_1+\epsilon_1,[x]_2+\epsilon_2)-H_\delta (x))H_\delta
                  ((-\infty,[x]_1)\!\times\! \mathbb{R} )^{ \lambda-1}+
\\
                  +\lambda (\lambda-1)(H_\delta
                  ([x]_1+\epsilon_1,[x]_2+\epsilon_2)-H_\delta (x))(H_\delta
                  ([x]_1+\epsilon_1,\infty)-H_\delta (-\infty,[x]_1)).
\end{multline*}
                  \normalsize This entails the inequality 
\begin{multline} 
\label{gsi}        
                  \frac{H_\delta ([x]_1+\epsilon_1,[x]_2+\epsilon_2)-H_\delta
                  (x)}{\epsilon_1\epsilon_2}\lambda
                  H_\delta((-\infty,[x]_1)\times \mathbb{R} )^{ \lambda-1}\le
\\
                  \le \frac{H^\star_\delta
                  ([x]_1+\epsilon_1,[x]_2+\epsilon_2)-H^\star_\delta
                  (x)}{\epsilon_1\epsilon_2}\le
\\
                  \le\frac{H_\delta
                  ([x]_1+\epsilon_1,[x]_2+\epsilon_2)-H_\delta
                  (x)}{\epsilon_1\epsilon_2} (\lambda
                  H_\delta((-\infty,[x]_1)\times \mathbb{R} )^{
                  \lambda-1}+
\\
                  +\lambda (\lambda-1)(H_\delta
                  ([x]_1+\epsilon_1,\infty)-H_\delta (-\infty,[x]_1)).
\end{multline}      
                  Because the leftmost part and the rightmost part of the
                  inequality (\ref{gsi}) have the same limit for
                  $(\epsilon_1,\epsilon_2)\stackrel{\epsilon_1,\epsilon_2>0}{
                  \longrightarrow }  (0,0)$, this is also the limit of the
                  middle part, 
\begin{multline} 
\label{gsd}        
                  \lim_{(\epsilon_1,\epsilon_2)\to (0,0)}\frac{H^\star_\delta
                  ([x]_1+\epsilon_1,[x]_2+\epsilon_2)-H^\star_\delta
                  (x)}{\epsilon_1\epsilon_2}=
\\
                  =\lim_{(\epsilon_1,\epsilon_2)\to
                  (0,0)}\frac{H_\delta
                  ([x]_1+\epsilon_1,[x]_2+\epsilon_2)-H_\delta
                  (x)}{\epsilon_1\epsilon_2}\lambda H_\delta
                  ((-\infty,[x]_1)\times \mathbb{R} )^{ \lambda-1}=
\\
                  =\lambda h_{ \delta } (x) H_{ \delta } ((-\infty,[x]_1)\times
                  \mathbb{R} )^{\lambda-1},
\end{multline}      
                  which means that $H^\star_\delta $ is absolutely continuous
                  with density given by (\ref{gds}).
\item             For $\delta,\delta'>0$, denote $L_{ \delta,\delta'}$ a layer
                  between the hyperplanes delimiting the half-spaces $L_\delta$
                  and $L_{\delta'}$. More precisely 
\begin{gather}
                  L_{ \delta,\delta'}=L_{\max\{\delta,\delta'\}} \setminus
                  L_{\min\{\delta,\delta'\}}. 
\end{gather}      
                  Let now $\delta>0$. The absolute continuity of $H$ together with
                  (\ref{ghl}) imply  
\begin{gather} 
\label{gll}        
                  \lim_{ \delta'\to \delta }H (L_{\delta'})=H (L_\delta)>0,
\end{gather}      
                  Consequently, $H (L_{\delta'})> \frac{1}{2} H (L_\delta)$ for
                  $\delta'$ close enough to $ \delta $, which for those
                  $\delta'$ entails  
\begin{gather} 
\label{ghn}        
                  \frac{ H_{ \delta' }((-\infty,[x]_1)\times \mathbb{R}
                  )^{\lambda-1}}{ H (L_{\delta'}) }\le\frac{1}{ H (L_{\delta'}) }\le
                  \frac{ 2}{ H (L_\delta) }.
\end{gather}      
                  Due to the continuity of $h$ and absolute continuity of $H$,
\begin{gather} 
\label{glh}        
                  \lim_{ \delta'\to \delta }\frac{ h}{ H (L_{\delta'}) }=\frac{
                  h}{ H (L_\delta) }, 
\end{gather}      
                  together with (\ref{ghn}) allowing to apply the dominated
                  convergence theorem, which yields
\begin{gather} 
\label{gld}        
                  (\forall A\in\mathcal{B}(\mathbb{R}^2))\;\lim_{ \delta'\to \delta
                  }\int_A\frac{ h}{ H (L_{\delta'}) }=\int_A\frac{ h}{ H (L_\delta) }.
\end{gather}      
                  In addition, (\ref{gdd}) and (\ref{ghn}) imply
\begin{multline} 
\label{gid}        
                  (\forall A\in\mathcal{B}(\mathbb{R}^2))\;|H_{\delta'}
                  (A)-H_\delta (A)|=|\int_A\frac{ h\mathbb{I}_{L_{\delta'}}}{ H
                  (L_{\delta'}) }-\int_A\frac{ h\mathbb{I}_{L_\delta}}{ H
                  (L_\delta) }|\le 
\\
                  \le|\int_A\frac{ h\mathbb{I}_{L_{\delta'}}}{ H
                  (L_{\delta'}) }-\int_A\frac{ h\mathbb{I}_{L_\delta}}{ H
                  (L_{\delta'}) }|+|\int_A\frac{ h\mathbb{I}_{L_\delta}}{ H
                  (L_{\delta'}) }-\int_A\frac{ h\mathbb{I}_{L_\delta}}{ H
                  (L_\delta) }|=
\\
                  =\frac{|\int_{A\cap L_{\delta'}} h-\int_{A\cap L_\delta}
                  h|}{H(L_{\delta'})}+|\int_{A\cap L_\delta}\frac{ h}{ H
                  (L_{\delta'}) }-\int_{A\cap L_\delta}\frac{ h}{ H (L_\delta)
                  }|\le
\\
                  \le \frac{2\int_{A\cap L_{ \delta,\delta'}
                  }h}{H(L_\delta)}+|\int_{A\cap L_\delta}\frac{ h}{ H (L_{\delta'})
                  }-\int_{A\cap L_\delta}\frac{ h}{ H (L_\delta) }|.
\end{multline}      
                  Together with (\ref{gld}) and with the absolute continuity of
                  $H$, this inequality leads for $A\in\mathcal{B}(\mathbb{R}^2)$ to 
\begin{gather} 
\label{gdl}        
                  \lim_{ \delta'\to \delta }H_{\delta'} (A)=H_\delta (A)+\lim_{
                  \delta'\to \delta }(H_{\delta'} (A)-H_\delta (A)) =H_\delta (A),
\end{gather}      
                  which combined with (\ref{ghn}) and (\ref{glh}) allows to
                  again apply the dominated convergence theorem, this time
                  yielding 
\begin{gather} 
\label{gls}        
                  \lim_{ \delta'\to \delta }\int_A\frac{ hH_{ \delta' }
                  ((-\infty,[x]_1)\times \mathbb{R} )^{\lambda-1}}{ H (L_{\delta'})
                  }=\int_A\frac{ hH_{ \delta }((-\infty,[x]_1)\times \mathbb{R}
                  )^{\lambda-1}}{ H (L_\delta) }. 
\end{gather}      
                  In a way analogous to (\ref{gid}), from (\ref{gds}) and
                  (\ref{ghn}) follows:
\begin{multline} 
\label{gis}        
                  (\forall A\in\mathcal{B}(\mathbb{R}^2))\;|H^\star_{\delta'}
                  (A)-H^\star_\delta (A)|\le\frac{2\int_{A\cap L_{
                  \delta,\delta'} }h}{H(L_\delta)}+
\\
                  +|\int_{A\cap L_\delta}\frac{ hH_{ \delta' }
                  ((-\infty,[x]_1)\times \mathbb{R} )^{\lambda-1}}{ H
                  (L_{\delta'}) }-\int_{A\cap L_\delta}\frac{ hH_{ \delta'
                  }((-\infty,[x]_1)\times \mathbb{R} )^{\lambda-1}}{ H
                  (L_\delta) }|.
\end{multline}    
                  Let $c_\delta:\mathbb{R}^2\to\mathbb{R}_+ $ be a function
                  defined 
\begin{gather} 
\label{gcd}        
                  (\forall x\in\mathbb{R}^2)\;c_\delta (x)=\delta-n^\top x
\end{gather}      
                  Taking into account the Borel-measurability of $c_\delta$, the
                  absolute continuity of $H$ and the fact that, due to
                  (\ref{gdd})--(\ref{gpk}),
\begin{gather} 
\label{gps}        
                  (\forall A\in\mathcal{B}(\mathbb{R}_+))\;P
                  (\delta,A)=H^\star_\delta (c_\delta^{-1} (A) )=\int_{c_\delta^{-1} (A)}
                  h^\star_\delta (x) dx ,  
\end{gather}      
                  (\ref{gls})--(\ref{gis}) already imply, for
                  $A\in\mathcal{B}(\mathbb{R}_+)$ the desired continuity of the
                  transition probability kernel with respect to $\delta $, 
\begin{multline}
                  \lim_{ \delta'\to \delta }P (\delta',A)=P (\delta,A)+\lim_{
                  \delta'\to \delta } (P (\delta',A)-P (\delta,A))=
\\
                  =P (\delta,A)+\lim_{ \delta'\to \delta } (H^\star_{\delta'}
                  (c_\delta^{-1} (A))-H^\star_\delta (c_\delta^{-1} (A))=P
                  (\delta,A).
\end{multline}      
\item             Let $\delta>0,A\in\mathcal{B}(\mathbb{R}_+)$ be such that
                  $\mu_+ (A)>0$. According to (\ref{gcd}), $c_\delta$ is
                  Lipschitz, therefore $c_\delta^{-1} (A)$ cannot be Lebesgue
                  negligible (\cite{fremlin10measure}, Proposition~262D). At the
                  same time, (\ref{ghl}) and (\ref{gcd}) imply $c_\delta^{-1}
                  (A) \subset L_\delta $, thus $h^\star_\delta$ is strictly
                  positive on $c_\delta^{-1} (A)$ due to the strict positivity
                  of $h$ and (\ref{gdd})--(\ref{gds}). Consequently,
\begin{gather} 
\label{gir}        
                  P (\delta,A)=H^\star_\delta (c_\delta^{-1} (A) )>0,
\end{gather}       
                  thus also $P ( (\exists k\in\mathbb{N})\; D_{t+k})\in
                  A|D_t=\delta )>0$, which
                  proves the $\mu_+$-irreducibility of $(D_t)_{t\in\mathbb{N}}$.

                  The aperiodicity of the chain will be proved by
                  contradiction. Let $\delta \in\mathbb{R}_+,\mathcal{D} \subset
                  \mathcal{B}(\mathbb{R}_+)$ be a cycle of
                  $(D_t)_{t\in\mathbb{N}}$, and $\Delta\in\mathcal{D}$ be such
                  that $\delta \in\Delta$. Then $\Delta$ is absorbing, and
                  therefore full with respect to  $\mu_+$ (\cite{meyn93markov},
                  Propositions~5.4.6 and~4.2.3). Suppose that $(D_t)_{t\in\mathbb{N}}$
                  is not aperiodic, thus there exists
                  $\Delta'\in\mathcal{D},\Delta'\ne\Delta$. The fact that
                  $\Delta$ is full with respect to  $\mu_+$ implies $\mu_+
                  (\Delta')=0$, which contradicts (\ref{gir}) with the choice
                  $A=\Delta'$, thus proving the aperiodicity.
\item             First, we will prove that the chain $(D_t)_{t\in\mathbb{N}}$
                  is weak Feller. Let $f:\mathbb{R}_+\to\mathbb{R}$ be
                  bounded continuous, thus $(\exists M_f>0) (\forall
                  x\in\mathbb{R}_+)|f (x)|<M_f$. The probability kernel
                  (\ref{gpk}) of $(D_t)_{t\in\mathbb{N}}$ induces a
                  transformation $\mathcal{T}_{(D_t)_{t\in\mathbb{N}}}$ of $f$
                  into the function
                  $\mathcal{T}_{(D_t)_{t\in\mathbb{N}}}f:\mathbb{R}_+\to\mathbb{R}$
                  defined  
\begin{gather} 
\label{gtf}        
                  (\forall \delta>0)\;\mathcal{T}_{(D_t)_{t\in\mathbb{N}}}f
                  (\delta )=\text{E} (fD_{t+1}|D_t=\delta ). 
\end{gather}      
                  Let now $\delta>0$. Due to (\ref{gdd}), (\ref{gds}),
                  (\ref{gl0}) and the bound on $f$, 
\begin{gather} 
\label{gfh}        
                  |f(c_\delta (x)) h^\star_\delta (x)|\le h (x) \frac{M_f}{H (L_0)},
\end{gather}      
                  which together with (\ref{gpk}) and (\ref{gps}) implies
\begin{multline} 
\label{gtb}        
                  |\text{E} (fD_{t+1}|D_t=\delta )|=|\int_{\mathbb{R}_+^2}f(c_\delta (x))
                  h^\star_\delta (x) dx |\le
\\
                  \le\int_{\mathbb{R}_+^2}|f(c_\delta (x))
                  h^\star_\delta (x)|dx\le \frac{M_f}{H (L_0) }.
\end{multline}      
                  From (\ref{gdd}), (\ref{gds}) and(\ref{gcd}), and from the
                  continuity of $f$ follows 
\begin{gather}
                  (\forall x\in\mathbb{R}_+^2)\; \lim_{ \delta'\to \delta
                  }f(c_{\delta'} (x)) h^\star_{\delta'} (x)=f(c_\delta (x))
                  h^\star_\delta (x),
\end{gather}      
                  which combined with (\ref{gfh}) allows to apply the dominant
                  convergence theorem, yielding 
\begin{gather}
                  \lim_{ \delta'\to \delta }\int_{\mathbb{R}_+^2}f(c_\delta (x))
                  h^\star_\delta (x)dx=\int_{\mathbb{R}_+^2}f(c_{\delta'} (x))
                  h^\star_\delta (x)dx.
\end{gather}      
                  Taking into account (\ref{gps}) implies
                  continuity of $\mathcal{T}_{(D_t)_{t\in\mathbb{N}}}f$ in
                  $\delta $, which in combination with (\ref{gtb}) means that
                  $\mathcal{T}_{(D_t)_{t\in\mathbb{N}}}$ transforms bounded
                  continuous functions into bounded continuous functions,
                  proving that $(D_t)_{t\in\mathbb{N}}$ is weak Feller.

                  Further, a consequence of the weak Feller property of
                  $(D_t)_{t\in\mathbb{N}}$, of its $\mu_+$-irreducibility, and
                  of the fact that the support of $\mu_+$ has non-empty
                  interior, is that all compacts in $\mathbb{R}_+$ are petite
                  sets (\cite{meyn93markov}, Proposition~6.2.8). In particular,
                  each set $(0,\beta]$ with $\beta>0$ is petite. 

                  Until now, we made no use of the additional conditions
                  (\ref{gmh}) and (\ref{gal}). First, from (\ref{gdd}),
                  (\ref{gds}), (\ref{gmh}), (\ref{gl0}) and the dominated
                  convergence theorem follows for $\alpha \in (0,\varepsilon )$,
\begin{multline} 
\label{gms}        
                  (\forall k\in\mathbb{N}) \;\int_{\mathbb{R}^2} |\alpha n^\top
                  x|^kh^\star_\delta (x)dx<+\infty\;\&\;\int_{\mathbb{R}^2}
                  \exp(|\alpha n^\top x|)h^\star_\delta
                  (x)dx=
\\
                  =\sum^{+\infty}_{k\in\mathbb{N}}\int_{\mathbb{R}^2} |\alpha
                  n^\top x|^kh^\star_\delta (x)dx<+\infty,  
\end{multline}
\begin{multline} 
                  (\forall k\in\mathbb{N}) \;\int_{\mathbb{R}^2} (\alpha n^\top x
                  )^kh (x)dx<+\infty \;\&\;\int_{\mathbb{R}^2} \exp(\alpha n^\top x
                  )h (x)dx=
\\
                  =\sum^{+\infty}_{k\in\mathbb{N}}\int_{\mathbb{R}^2}
                  (\alpha n^\top x)^kh (x)dx<+\infty, 
\end{multline}
\begin{multline} 
\label{gme}        
                  (\forall k\in\mathbb{N}) \;\int_{\mathbb{R}^2} (\alpha n^\top
                  x)^kh^\star_\delta (x)dx<+\infty\;\&\;\int_{\mathbb{R}^2}
                  \exp(\alpha n^\top x)h^\star_\delta
                  (x)dx=
\\
                  =\sum^{+\infty}_{k\in\mathbb{N}}\int_{\mathbb{R}^2} (\alpha
                  n^\top x)^kh^\star_\delta (x)dx<+\infty,
\end{multline}      
                  whereas due to (\ref{gal}), it is possible to introduce 
\begin{gather} 
\label{gdi}        
                  \delta_\infty=\int_{\mathbb{R}^2}n^\top
                  xH((-\infty,[x]_1)\times \mathbb{R} )^{ \lambda-1}h(x) dx>0. 
\end{gather}      
                  From (\ref{gdd}), (\ref{gds}) and (\ref{gl0}) for
                  $x\in\mathbb{R}^2$ follows
\begin{gather}
                  \lim_{ \delta \to+\infty }n^\top xh^\star_\delta (x)=n^\top
                  x(x)H((-\infty,[x]_1)\times \mathbb{R} )^{ \lambda-1}h(x),  
\\
                  \text{and } |n^\top xh^\star_\delta (x)|\le \frac{|n^\top x|h(x)}{H (L_0) },
\end{gather}      
                  which combined with the finite 1st moment of $H$ allows to apply the
                  dominated convergence theorem, yielding   
\begin{gather}
                  \lim_{ \delta \to+\infty }\int_{\mathbb{R}^2}n^\top
                  xh^\star_\delta (x)=\delta_\infty.  
\end{gather}      
                  This together with (\ref{gdi}) entails the existence of
                  $\beta>0$ such that  
\begin{gather} 
\label{gbe}        
                  (\forall \delta>\beta)\;\int_{\mathbb{R}^2}n^\top xh^\star_\delta (x)\in
                  (\frac{2}{3}\delta_\infty,\frac{4}{3}\delta_\infty ). 
\end{gather}      
                  For $\alpha \in\mathbb{R}_+$, define the function $\Delta
                  V_\alpha $ on $\mathbb{R}_+$ by  
\begin{gather}
                  (\forall \delta>0)\;\Delta V_\alpha (\delta )=\text{E}
                  (V_\alpha (D_{t+1})|D_t=\delta)-V_\alpha (\delta ).  
\end{gather}      
                  Then for $\delta>0$, (\ref{gyt}), (\ref{gxt}), (\ref{gdt}) and (\ref{gva})
                  lead to 
\begin{multline}
                  \Delta V_\alpha (\delta )=\text{E} \exp (\alpha
                  (\delta-n^\top M_t^{i_t,j_{t,i_t} (\delta)}))-V_\alpha (\delta )=
\\
                  =\text{E}(V_\alpha (\delta )\exp (-\alpha n^\top M_t^{i_t,j_{t,i_t}
                  (\delta)})-V_\alpha (\delta )=
\\
                  =V_\alpha (\delta )\text{E}\exp (-\alpha n^\top M_t^{i_t,j_{t,i_t}
                  (\delta)})-V_\alpha (\delta )=
\\
                  =V_\alpha (\delta )\int_{\mathbb{R}^2}\exp(-\alpha n^\top
                  x)h^\star_\delta (x)dx-V_\alpha (\delta ).
\end{multline}   
                  In particular for $0<\delta \le\beta$ from (\ref{gdd}),
                  (\ref{gds}), (\ref{gmh}) and (\ref{gl0}) follows for $\alpha\in
                  (0,\varepsilon )$
\begin{gather} 
\label{glb}        
                  \Delta V_\alpha (\delta )\le\frac{V_\alpha (\delta )}{H (L_0)
                  }\int_{\mathbb{R}^2}\exp(|\alpha n^\top x|)h (x)dx-V_\alpha
                  (\delta ). 
\end{gather}    
                  On the other hand, for $\delta>\beta$ and $\alpha\in
                  (0,\varepsilon )$, (\ref{gdd}), (\ref{gds}), (\ref{gmh}),
                  (\ref{gme}) and (\ref{gbe}) lead to
\begin{multline} 
\label{ggb}        
                  \Delta V_\alpha (\delta )=V_\alpha (\delta
                  )\int_{\mathbb{R}^2} (1-\alpha n^\top x+\sum_{k=2}^\infty
                  \frac{(-\alpha n^\top x)^k }{k!} )h^\star_\delta (x)dx-V_\alpha (\delta )=
\\
                  =\alpha V_\alpha (\delta)(-\int_{\mathbb{R}^2} n^\top xh^\star_\delta
                  (x)dx+\alpha\sum_{k=2}^\infty \frac{
                  \alpha^{k-2}\int_{\mathbb{R}^2}(- n^\top x)^kh^\star_\delta
                  (x)dx}{k!}  \le
\\
                  \le\alpha V_\alpha
                  (\delta)\left (-\frac{2}{3}\delta_\infty+\alpha\sum_{k=2}^\infty
                  \frac{ \alpha^{k-2}\int_{\mathbb{R}^2}|n^\top x|^kh(x)dx}{k!} \right). 
\end{multline}    
                  For $0\le\alpha<\alpha'<\varepsilon $ is 
\begin{gather}
                  0\le\alpha\sum_{k=2}^\infty \frac{
                  \alpha^{k-2}\int_{\mathbb{R}^2}|n^\top
                  x|^kh(x)dx}{k!}\le\alpha\sum_{k=2}^\infty \frac{
                  \alpha'^{k-2}\int_{\mathbb{R}^2}|n^\top x|^kh(x)dx}{k!},
\end{gather}      
                  which together with $\lim_{ \alpha \to
                  0}\alpha\sum_{k=2}^\infty \frac{
                  \alpha'^{k-2}\int_{\mathbb{R}^2}|n^\top x|^kh(x)dx}{k!}=0$
                  entails
\begin{gather}
                  \lim_{ \alpha \to
                  0}\alpha\sum_{k=2}^\infty \frac{
                  \alpha^{k-2}\int_{\mathbb{R}^2}|n^\top x|^kh(x)dx}{k!}=0.  
\end{gather}      
                  Consequently,
\begin{gather}
                  (\exists \alpha_0>0) (\forall \alpha \in (0,\alpha_0) )\;\lim_{\alpha \to
                  0}\alpha\sum_{k=2}^\infty \frac{
                  \alpha^{k-2}\int_{\mathbb{R}^2}|n^\top
                  x|^kh(x)dx}{k!}<\frac{1}{3}\delta_\infty . 
\end{gather}   
                  Hence, putting $ \alpha \in (0,\min\{\alpha_0,\varepsilon \})$
                  into (\ref{ggb}) yields
\begin{gather} 
\label{g13}        
                  \Delta V_\alpha (\delta )\le-\frac{1}{3}\alpha V_\alpha (delta).
\end{gather}      
                  Let $\alpha_n=\min\{\alpha_0,\varepsilon,3 \}$. Then combining
                  (\ref{glb}) and (\ref{g13}) leads for $\alpha \in
                  (0,\alpha_n)$ to  
\begin{gather}
                  (\forall \delta>0)\;\Delta V_\alpha (\delta
                  )\le-\frac{1}{3}\alpha V_\alpha (delta)+\frac{V_\alpha (\delta )}{H (L_0)
                  }\int_{\mathbb{R}^2}\exp(|\alpha n^\top x|)h (x)dx \mathbb{I}_{(0,\beta]} (x),
\end{gather}      
                  which according to the Geometric Ergodic Theorem
                  (\cite{meyn93markov}, Theorem 15.0.1) proves
                  $(D_t)_{t\in\mathbb{N}}$ to be $V_\alpha-$geometrically
                  ergodic. \qed
\end{enumerate}   
\end{proof}
\begin{example}
                  In \cite{chotard14markov}, movements with Gaussian
                  distributions with zero mean are considered. If $H$ is a
                  non-degenerated 2-dimensional Gaussian distribution with
                  $\int_xh (x)dx=0$, then all assumptions of
                  Proposition~\ref{g2} are fulfilled, including the additional
                  conditions (\ref{gmh}) and (\ref{gal}):
\begin{romanitems}
\item             $h$ is continuous strictly positive.
\item             $n^\top M_t^{i,j}$ has a 1-dimensional Gaussian distribution,
                  thus all its moments are finite and the moment-generating
                  function $\text{E} \exp(\alpha n^\top M_t^{i,j}) $ is finite
                  for all $\alpha \in\mathbb{R}$. From this, the finiteness of
                  $\text{E} \exp(|\alpha n^\top M_t^{i,j}|)$ follows due to
                  the symmetry of the Gaussian density, the finiteness of $k$-th
                  moments of $|n^\top M_t^{i,j}|$ is for $k$ even already
                  equivalent to the finiteness of $k$-th moments of
                  $n^\top M_t^{i,j}$, and for $k$ odd is a consequence of the
                  inequality 
\begin{gather}
                  |n^\top M_t^{i,j}|^k\le 
\begin{cases}
                  1&\text{ if } |n^\top M_t^{i,j}|\le 1,
\\
                  |n^\top M_t^{i,j}|^{k+1}&\text{ else }.
\end{cases}       
\end{gather}      
\item             Denote $H_1,H_2$ the marginal distributions of $H$, and
                  $h_1,h_2$ their respective marginal densities. Then the
                  symmetry of the Gaussian denisty implies
\begin{multline}
                  \int_{\mathbb{R}^2}n^\top xH((-\infty,[x]_1)\times \mathbb{R}
                  )^{ \lambda-1}h(x) dx=
\\
                  =\int_{\mathbb{R}}\int_{\mathbb{R}}h(x)d[x]_2[n]_1[x]_1H_1((-\infty,[x]_1)
                  )^{ \lambda-1}d[x]_1+
\\
                  +\int_{\mathbb{R}}\int_{\mathbb{R}}H_1((-\infty,[x]_1) )^{
                  \lambda-1}h(x)dx[x]_1[n]_2[x]_2d[x]_2
\\
                  =[n]_1\int_{\mathbb{R}}h_1([x]_1)[x]_1H_1((-\infty,[x]_1))d[x]_1+
\\
                  [n]_2\!\int_{\mathbb{R}_+}\!
                  (\!\int_{\mathbb{R}}\!H_1((-\infty,[x]_1))^{ \lambda-1}
                  -\!\int_{\mathbb{R}}\!H_1((-\infty,-[x]_1))^{
                  \lambda-1})h(x)d[x]_1)[x]_2d[x]_2
\\
                  =[n]_1\int_{\mathbb{R}_+}h_1([x]_1)[x]_1H_1((-[x]_1,[x]_1))d[x]_1+
\\
                  +[n]_2\int_{\mathbb{R}_+}\int_{
                  ((-[x]_1,[x]_1))}H_1((-\infty,[x]_1))^{ \lambda-1}h(x)d[x]_1[x]_2d[x]_2>0.
\end{multline}        
\end{romanitems}   
\end{example}

\section{         Investigation of Movement Distributions by Means of Copulas}

                  To get a deeper understanding of the influence that the
                  distribution of the movements has at the resulting Markov
                  chain, it is advantageous to decompose that distribution into
                  its marginals and the copula combining them. In this section,
                  we express the density $h$ of the movements distribution
                  using such a decomposition. We will pay a separate attention
                  to the particularly well transparent structure of Archimedean
                  copulas. The results, formulated below in
                  Proposition~\ref{c3}, are essentially an application of two
                  important theorems of the copula theory (first of them being
                  the famous Sklar's theorem \cite{sklar59fonctions} that
                  established the relationship between multivariate
                  distributions and copulas).
\begin{theorem}
                  {\bfseries Sklar} \cite{sklar59fonctions} Let $m\in
                  \mathcal{N}$ and $F_1\dots F_m$ be distribution functions of
                  one-dimensional random variables. For $i=1\dots m$, let
                  $\val F_i$ denote the value set of $F_i$ and $F^-_i$ the
                  pseudoinverse of $F_i$, defined
\begin{gather}
                  (\forall y\in[0,1])\;F^-_i (y)=\inf\{x\in\mathbb{R}:F_i (x)\ge y \}.
\end{gather}      
                  Then there exists an $m$-dimensional distribution function\,$F$
                  such\,that $F_1\dots F_m$ are marginals of $F$ if and only if
                  there exists an $m$-dimensional copula $C$, i.e., a
                  distribution function on $[0,1]^m$ with uniform marginals,
                  fulfilling  
\begin{gather} 
\label{csk}        
                  (\forall x\in\mathbb{R}^m)\;F (x)=C (F_1 ([x]_1)\dots F_m ([x]_m) ).  
\end{gather}      
                  In the positive case, $C$ is uniquely determined on the set
                  $\val F_1\times\dots\times\val  F_m$, and is given by 
\begin{gather}
                  (\forall  u\in\val  F_1\times\dots\times\val  F_m)\;
                  C (u)=F (F^-_1 ([u]_1)\dots F^-_m ([u]_m) ).
\end{gather}     
\end{theorem}
\begin{theorem} 
\label{cm}        
                  \cite{mcneil09multivariate} Let $m\in
                  \mathcal{N},\psi:[0,+\infty]\to[0,1]$ be an Archimedean
                  generaor, i.e., $\psi (0)=1,\psi (+\infty )=\lim_{t\to +\infty
                  }\psi (t)=0, \psi$ is continuous and strictly decreasing on
                  $[0,\inf\{t:\psi (t)=0 \} )$, and let
                  $\psi^{-1}:[0,1]\to[0,+\infty]$ be defined
\begin{gather}
                  (\forall u\in [0,1])\;\psi^{-1} (u)=\inf\{t:\psi (t)=u \}
\end{gather}      
                  and $C_\psi:[0,1]^m\to[0,1]$ be defined
\begin{gather} 
\label{car}        
                  (\forall u\in[0,1]^m)\;C_\psi (u) =\psi
                  (\psi^{-1} ([u]_1)+\dots+\psi^{-1} ([u]_m)).  
\end{gather}      
                  Then 
\begin{enumerate}
\item             $C_\psi$ is a copula if and only if $\psi$ is $m$-monotone,
                  i.e., $\psi$ is continuous on $[0,+\infty],\psi^{ (k)
                  }$ exists on $\mathbb{R}_+$ for $k=1\dots m-2$, $\psi^{ (m-2) }$ is
                  decreasing and convex on $\mathbb{R}_+$, and $(\forall
                  k\in\{0\dots m-2\}) (\forall t\in\mathbb{R}_+)\; (-1)^k\psi^{ (k) }
                  (t)\ge 0$.
\item             In particular, a sufficient condition for $C_\psi$ to be a
                  copula is $\psi$ being completely monotone, i.e., $m$-monotone
                  for each $m\in \mathcal{N},m\ge 2$. 
\end{enumerate}   
\end{theorem}
\begin{proposition} 
\label{c3}        
                  Let $H$ be the distribution from Proposition~\ref{g2}, $H_1$ and
                  $H_2$ be its marginals, and $C$ be a copula relating $H$ to $H_1$ and
                  $H_2$ according to the Sklar's theorem. Then the following holds:
\begin{enumerate}
\item             Sufficient for $H$ to have a continuous strictly
                  positive density is the simultaneous validity of the following
                  three conditions.
\begin{romanitems}
\item             $H_1$ and $H_2$ have continuous strictly positive densities
                  $h_1$ and $h_2$, respectively.
\item             $C$ has a continuous strictly positive density $c$.
\item             for $x\in\mathbb{R}^2$, 
\begin{gather} 
\label{chg}        
                  h (x)=c (H_1 ([x]_1),H_2 ([x]_2) )h_1([x]_1)h_2 ([x]_2).
\end{gather}      
\end{romanitems}   
\item             If $C$ is Archimedean, i.e., $C=C_\psi$ according to
                  (\ref{car}) with some Archimedean generator $\psi$, then it is
                  sufficient to replace (ii) and (iii), respectively, with 
\begin{enumerate}
\item[(ii')]      $\psi$ is at least 4-monotone.
\item[(iii')]     for $x\in\mathbb{R}^2$,
\begin{gather} 
\label{cha}        
                  h (x)=\frac{\psi'' (\psi^{-1} (H_1 ([x]_1))+\psi^{-1} (H_2
                  ([x]_2)) ) }{\psi' (\psi^{-1} (H_1 ([x]_1))+\psi^{-1} (H_2
                  ([x]_2)) )}h_1([x]_1)h_2 ([x]_2). 
\end{gather}      
\end{enumerate}   
\end{enumerate}   
\end{proposition}
\begin{proof}
\item             To prove (\ref{chg}), the relationships 
\begin{gather} 
\label{cpd}        
                  (\forall x\in \mathbb{R}^2)\;h (x)=\frac{\partial^2
                  H}{\partial [x]_1 \partial [x]_2 } (x)
                  ,h_1([x]_1)=\frac{H_1 }{d[x]_1} ([x]_1),h_2([x]_2)=\frac{H_2
                  }{d[x]_2} ([x]_2), 
\end{gather}      
                  are combined with (\ref{csk}) and $c (u)=\frac{\partial^2
                  C}{\partial [u]_1 \partial [u]_2 } (u)$, making use of
                  the assumption that $h_1,h_2$ and $c$ are continuous.
\item             To prove (\ref{cha}), the relationships (\ref{cpd}) are
                  combined with (\ref{car}) and (\ref{chg}), making use of
                  the assumption that $h_1,h_2$ and $c$ are continuous, and of
                  the fact that if $y=\psi (t)$, where $t=\psi^{-1}
                  (u_1)+\psi^{-1} (u_2)$, then 
\begin{gather}
                  (\psi^{-1})' (y)=\frac{1}{\psi (t) } =\frac{1}{\psi^{-1}
                  (u_1)+\psi^{-1} (u_2)}.
\end{gather}      
                  \qed
\end{proof}
\begin{example}
                  A simple Archimedean generator is given by 
\begin{gather} 
\label{cgg}        
                  (\forall t\in \mathbb{R}_+)\; \psi (t)=\exp (-t^{ \frac{1}{
                  \vartheta } }) \text{, or equivalently, } (\forall y\in (0,1)
                  )\; \psi^{-1} (y)= (-\ln y)^\vartheta,
\end{gather}      
                  with $\vartheta \in[1,+\infty)$. From (\ref{cgg}) follows that 
\begin{gather}
                  (\forall t\in \mathbb{R}_+)\; \psi' (t)=-\frac{\exp (-t^{ \frac{1}{
                  \vartheta } })}{ \vartheta t^\frac{ \vartheta-1}{ \vartheta } }<0,
\end{gather}      
                  and that the sign of $\psi^{ (k) }$ is opposite to the sign of
                  $\psi^{ (k-1) }$ for $k\ge 2$. Thus  $\psi $ is completely
                  monotone and using (\ref{car}),  $\psi$ indeed defines a copula,
\begin{gather} 
\label{gu}        
                  C_\psi (u)=\exp \left (-\left ( (-\ln
                  [u]_1)^\vartheta+(-\ln [u]_2)^\vartheta \right )^{ \frac{1}{
                  \vartheta }}\right  ).                   
\end{gather}      
                  Such copulas are called \emph{Gumbel copulas}
                  \cite{nelsen06introduction}. The particular choice
                  $\vartheta=1$ yields the \emph{product copula}, 
\begin{gather}
                  \Pi(u_1\dots u_m)=\exp (- ( (-\ln u_1)+(-\ln u_2) ) 
                  )=\prod_{i=1}^mu_i, 
\end{gather}      
                  which describes the independence of marginals because due to
                  (\ref{csk}), 
\begin{gather}
                  F (x_1\dots x_m)=\prod_{i=1}^mF_i (x_i).
\end{gather}      
                  To investigate movement distributions, any of (\ref{chg}) or
                  (\ref{cha}) lead to the following expression for the density $h$
                  needed in Proposition~\ref{g2}.
\begin{multline}
                  h (x)=H (x)\frac{h_1 ([x]_1)h_2 ([x]_2)}{H_1 ([x]_1)H_2 ([x]_2)
                  }\left (\frac{(-\ln H_1 ([x]_1))(-\ln H_2 ([x]_2))}{
                  \left ( (-\ln H_1 ([x]_1))^\vartheta+(-\ln H_2
                  ([x]_2))^\vartheta \right )^{ \frac{2}{ \vartheta } } } \right
                  )^{ \vartheta-1} \cdot 
\\
                  \cdot  \left (1-\frac{ \vartheta-1}{ \vartheta\left ( (-\ln
                  H_1 ([x]_1))^\vartheta+(-\ln H_2 ([x]_2))^\vartheta \right ) }
                  \right).
\end{multline}    
\end{example}


%

\bibliography{odkazy}

\end{document}